\documentclass[a4paper,11pt]{article}
\usepackage{graphicx,amssymb,amsmath,amsthm,textcomp}
\input{psfig.sty}
\usepackage{color}
\usepackage{appendix}
\newtheorem{theorem}{Theorem}
\newtheorem{lemma}{Lemma}

\newtheorem{definition}{Definition}
\newtheorem{observation}{Observation}

\newtheorem{fact}{Fact}

\newcommand{\remove}[1]{}

\newcommand{\IR}{\mathbb{R}}

\usepackage{amsfonts}

\title{Corrigendum to: ``Linear time algorithm to cover and hit a set of
line segments optimally by two axis-parallel squares'', Theoretical Computer Science 769 (2019) 63--74}
\author{Sanjib Sadhu$^1$ \and Xiaozhou He$^2$ \and Sasanka Roy$^3$
\and  Subhas C. Nandy$^3$ \and Suchismita Roy$^1$}
\date{$^1$Dept. of CSE, National Institute of Technology Durgapur, India\\
$^2$ Business School, Sichuan University, Chengdu, China\\
$^3$Indian Statistical Institute, Kolkata, India}

\begin{document}

\maketitle

{\bf Abstract:} 
In the paper ``Linear time algorithm to cover and hit a set of
line segments optimally by two axis-parallel squares'', TCS Volume 769 (2019), pages 63--74, the LHIT problem is proposed as follows: 
\begin{description}
\item[] For a given set of non-intersecting line segments 
${\cal L} = \{\ell_1, \ell_2, \ldots, \ell_n\}$ in $I\!\!R^2$, compute two axis-parallel congruent 
squares ${\cal S}_1$ and ${\cal S}_2$ of minimum size whose union hits all the line segments in $\cal L$,
\end{description}
and a linear time algorithm was proposed. Later it was observed that the algorithm has a bug. In this corrigendum, we corrected the algorithm. The time complexity of the corrected algorithm is 
$O(n^2)$. 


\section{Introduction}
For a given set of line segments ${\cal L} = \{\ell_1, \ell_2, \ldots, \ell_n\}$ in $I\!\!R^2$, 
 the following two problems were proposed in \cite{sadhu}:
 \begin{description}
\item[Line segment covering (LCOVER) problem:] Given a set ${\cal L}=
\{\ell_1,\ell_2,\ldots,\ell_n\}$ of $n$ line segments (possibly intersecting) 
in $\IR^2$, compute two  congruent 
squares ${\cal S}_1$ and ${\cal S}_2$ of minimum size whose union covers all the members in $\cal L$.
\item[Line segment hitting (LHIT) problem:] Given a set ${\cal L}=
\{\ell_1,\ell_2,\ldots,\ell_n\}$ of $n$ non-intersecting line segments 
in $\IR^2$, compute two axis-parallel congruent 
squares ${\cal S}_1$ and ${\cal S}_2$ of minimum size whose union hits all the line segments in $\cal L$.
 \end{description}

 For both the problems, linear time algorithms were proposed. 
 Later, we identified that there is a bug in the proposed algorithm for the LHIT problem. 
In this corrigendum, we present a revised algorithm for the LHIT problem.
The time complexity of this algorithm is $O(n^2)$ in the worst case.

\begin{figure}
\centering
\includegraphics[width=0.7\textwidth]{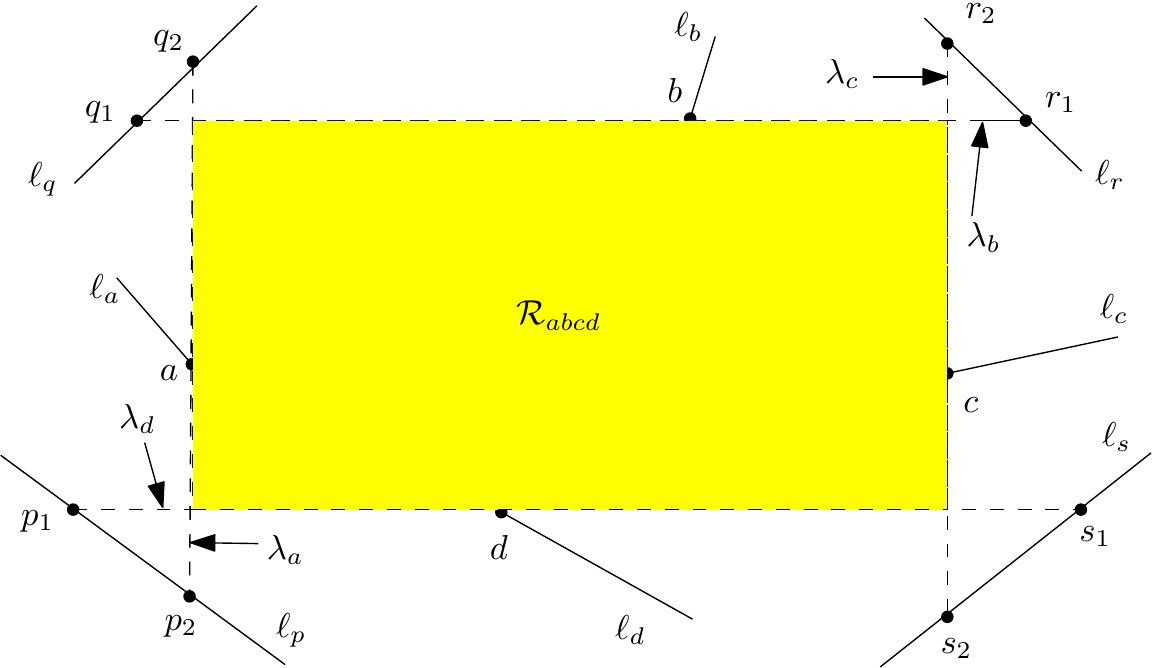}
\caption{The axis parallel rectangle ${\cal R}_{abcd}$ defined by the
points $a$, $b$, $c$ and $d$ that does not hit all the members in ${\cal L}$.}
\label{degenerated_fig1}
\end{figure} 

{\em An axis parallel rectangle $\cal R$ is called a {\it hitting rectangle} if every member
in $\cal L$ is either intersected by ${\cal R}$ or is completely contained in ${\cal R}$.}
In \cite{sadhu}, we performed a linear scan among the objects in $\cal L$ to identify four 
points $a$, $b$, $c$ and $d$, where $a$ is the right end-point of 
a segment $\ell_a \in \cal L$ having minimum $x$-coordinate, $b$ is the bottom end-point 
of a segment $\ell_b \in \cal L$ having maximum $y$-coordinate, $c$ is the left end-point 
of a segment $\ell_c \in \cal L$ having maximum $x$-coordinate, and $d$ is the top 
end-point of a segment $\ell_d\in \cal L$ having minimum $y$-coordinate (see Figure~\ref{degenerated_fig1}). The axis-parallel 
rectangle  whose ``left'', ``top'', ``right'' and ``bottom'' sides contain  the points $a$, 
$b$, $c$ and $d$ respectively, is denoted by ${\cal R}_{abcd}$. In \cite{sadhu}, we claimed that this 
axis-parallel rectangle ${\cal R}_{abcd}$ is a {\it hitting rectangle}. Using this rectangle, 
we computed two congruent squares of minimum size that hits all the line segments in $\cal L$. 
Later, we observed that ${\cal R}_{abcd}$ is not always a hitting rectangle (see Figure 
\ref{degenerated_fig1}). Thus, the proposed algorithm for the LHIT problem may fail in some 
pathological cases. In this corrigendum, we correct our mistake. As in \cite{sadhu}, 
we first compute ${\cal R}_{abcd}$. If it hits all the segments in $\cal L$, our 
proposed linear time algorithm in \cite{sadhu} will work for the LHIT problem. However, if 
${\cal R}_{abcd}$ does not hit all the segments in $\cal L$, we propose an $O(n^2)$ time 
algorithm for the LHIT problem.

As mentioned earlier, the members in $\cal L$ are non-intersecting. We use the following notations 
to describe our revised algorithm. Here, $\lambda_a$, $\lambda_b$, $\lambda_c$ and $\lambda_d$ denote  
the lines containing the left, top, right and bottom boundaries of ${\cal R}_{abcd}$ respectively.
Let $\ell_p$ be the segment which is not hit by ${\cal R}_{abcd}$ and lies farthest from both 
``$a$'' and ``$d$'' along vertically downward and horizontally leftward directions respectively.
Similarly the other segments $\ell_q$, $\ell_r$ and $\ell_s$ are defined (see 
Figure~\ref{degenerated_fig1}). Let $(p_1,p_2)$ be the two points of intersection 
of $\ell_p$ with $\lambda_a$ and $\lambda_d$ 
respectively. Similarly the point-pairs $(q_1,q_2)$, $(r_1,r_2)$ and
$(s_1,s_2)$ are defined (see Figure \ref{degenerated_fig1}).
Note that, all the segments $\ell_p, \ell_q, \ell_r, \ell_s$ may not exist. 
However, if at least one of these four 
segments exists, then our proposed algorithm in \cite{sadhu} will fail.

We first propose an algorithm for computing a minimum sized axis parallel square $\cal S$ 
that hits a given set of line segments $\cal L$. We use this result to compute 
the two axis parallel congruent squares ${\cal S}_1$ and
${\cal S}_2$ of minimum size for hitting all the segments in $\cal L$.

\section{One hitting square}\label{one-hit}

\begin{figure}[t]
\begin{minipage}[b]{0.5\linewidth}
\centering
\includegraphics[width=0.7\textwidth]{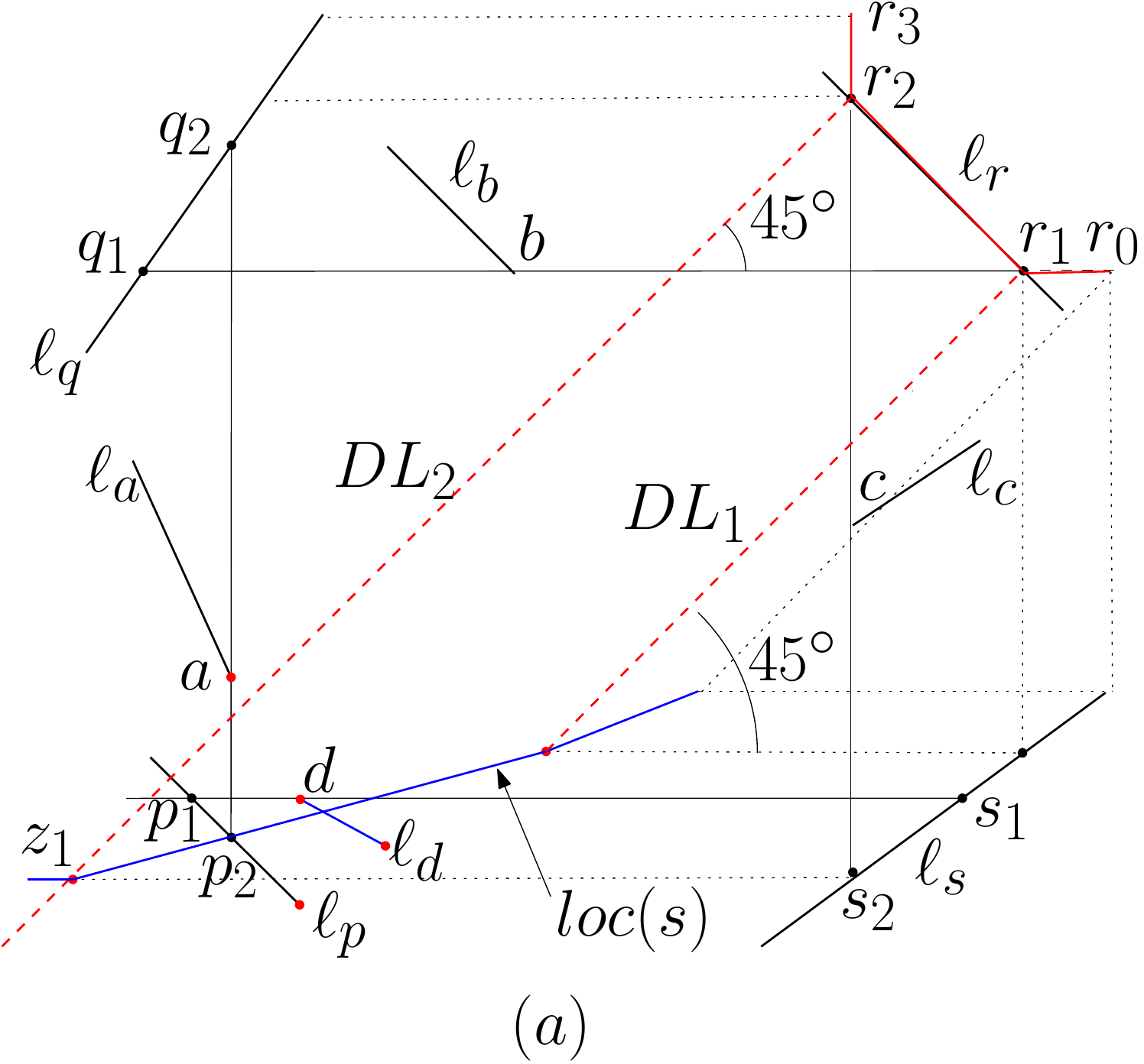}
\end{minipage} 
\hspace{0.005cm}
\begin{minipage}[b]{0.5\linewidth}
\centering
\includegraphics[width=.8\textwidth]{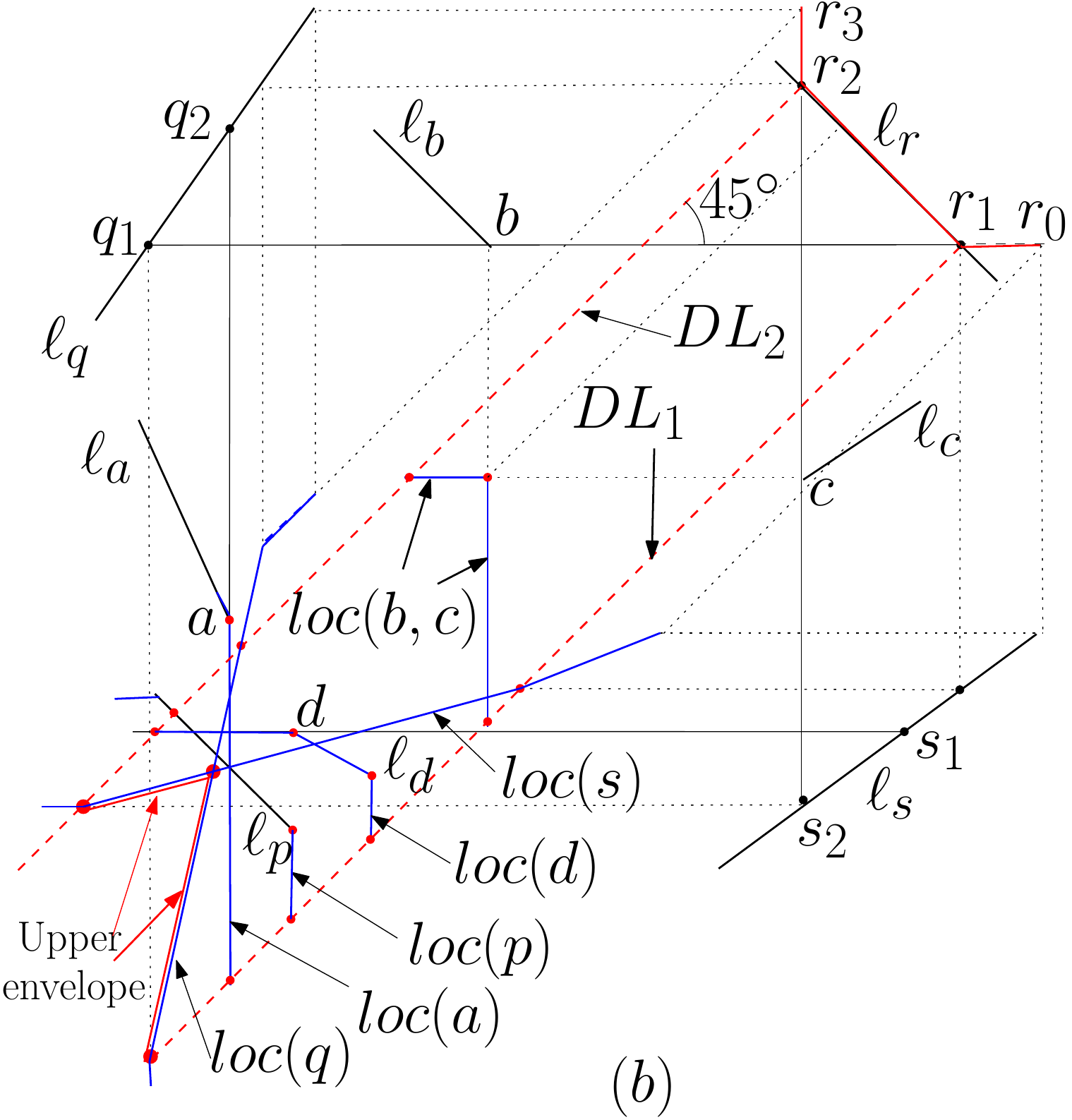}
\end{minipage}  
\caption{(a) Computation of $loc(s)$ 
(b) Computation of a minimum sized axis parallel square that hits all the segments}
\label{degenerated_fig11}
\end{figure}

\begin{fact}
\label{simple}
A square, that hits $\ell_a$, $\ell_b$, $\ell_c$, $\ell_d$, $\ell_p$, $\ell_q$, $\ell_r$
 and $\ell_s$ (those which exists), will hit all the segments in $\cal L$.
\end{fact}

\begin{proof}
Let $\cal R$ be a square that hit all the segments in $\{\ell_a,\ell_b,\ell_c,\ell_d,\ell_p,\ell_q,\ell_r, 
\ell_s\}$, and $\ell \in {\cal L}\setminus \{\ell_a,\ell_b,\ell_c,\ell_d,\ell_p,\ell_q,\ell_r, 
\ell_s\}$ be a segment that is not hit by $\cal R$. 
The square $\cal R$ must cover ${\cal R}_{abcd}$ (Figure~\ref{degenerated_fig1}). So by our assumption, 
$\ell$ must not intersect ${\cal R}_{abcd}$. From the definition of the distinguished points
``$a$'', ``$b$'', ``$c$'' and ``$d$'', the segment $\ell$ must intersect
both the members of at least one of the tuples 
$(\lambda_a,\lambda_b)$, $(\lambda_b,\lambda_c)$ and 
$(\lambda_c,\lambda_d)$, and $(\lambda_a,\lambda_d)$ outside ${\cal R}_{abcd}$. Without loss of 
generality, assume that $\ell$ hits $(\lambda_a,\lambda_d)$.
In order to hit $\ell_p$ by $\cal R$, it must hit $\ell$. Thus,
we have the contradiction. 
\end{proof}

{\bf Implication of Fact \ref{simple}:} The minimum size square hitting all
the segments in a given set $\cal L$ is defined by 
at most eight segments $\{\ell_a,\ell_b,\ell_c,\ell_d,\ell_p,\ell_q,\ell_r, 
\ell_s\}$ of $\cal L$.

\begin{observation}
\label{basic}
(i) The subset of $\cal L$ defining the possible minimum size squares hitting 
all the segments in $\cal L$ (if more than one such squares exist) is unique.\\
(ii) If $\cal S$ is the {\em minimum sized axis parallel square}  that hits all the line segments in 
$\cal L$, then at least one of the vertices of $\cal S$ will lie on one of the four segments 
$\overline{p_1p_2}$, $\overline{q_1q_2}$, $\overline{r_1r_2}$ and $\overline{s_1s_2}$.
\end{observation}
\begin{proof}
{\bf \boldmath part (i)}: A minimum sized square ${\cal S}$ hitting all the segments is defined by either two or three 
segments which are termed as the defining segments for ${\cal S}$.\\
(a) If the number of defining segments of ${\cal S}$ is two, then those two segments must
touch the two opposite boundaries (left, right) or (top, bottom) of ${\cal S}$,
or two  diagonal vertices of ${\cal S}$. The defining segments must touch the
boundary of square ${\cal S}$ externally i.e. from outside,
otherwise ${\cal S}$ can be further reduced.
\begin{itemize}
 \item {\bf \boldmath Two defining segments touch the two opposite sides of the square ${\cal S}$}:\\
 Here, the maximum of ``minimum horizontal
distance'' and ``minimum vertical distance'' between ``two defining segments''
(say $\ell_1$ and $\ell_2$) will be the length
of the side of ${\cal S}$. See Figure~\ref{gruent_case_label}(a,b).
If there exists another square ${\cal S'}$ that hits all 
the segment, then ${\cal S'}$ will also
hit $\ell_1$ and $\ell_2$ indicating that the horizontal/vertical
span will increase or remain at least same as that of $\cal S$.
If $\cal S$ and $\cal S'$ are of same size (see 
Figure~\ref{gruent_case_label}(a,b)), then the defining segments of $\cal S$
and $\cal S'$ are same.
\item {\bf \boldmath Two defining segments touch the two diagonal vertices of the square ${\cal S}$}:\\
If ${\cal S}$ is defined by two segments $\ell_1$ and $\ell_2$ touching 
its two diagonal vertices, then the segments are either 
parallel to each other (see Figure~\ref{gruent_case_label}(c)) or 
the minimum distance between two defining segments $\ell_1$ and $\ell_2$ is
the length of diagonal of ${\cal S}$ (See Figure~\ref{gruent_case_label}(d)).  
Here also if there exists another square ${\cal S}'$ 
defined by other two segments $(\ell'_1,\ell'_2)\neq (\ell_1,\ell_2)$
then the horizontal/vertical
span will increase or remain at least same as that of $\cal S$.
If $\cal S$ and $\cal S'$ are of same size (in case $\ell_1$
and $\ell_2$ are parallel as shown in Figure~\ref{gruent_case_label}(c)),
then the defining segments of $\cal S$
and $\cal S'$ are same.
\end{itemize}
(b) If the number of defining segments of ${\cal S}$ are three, 
say $\ell_1$, $\ell_2$ and $\ell_3$, then two of them 
must touch the two opposite boundaries (left, right) or (top, bottom)
of the square ${\cal S}$. If there exists any square $\cal S'$
that hits all the segments in $\cal L$, then
arguing as in the earlier case,
it can be shown that the size of ${\cal S'}$ is at least as large as $\cal S$,
and the defining segments will remain same.

{\bf \boldmath Part (ii)}: Assume that none of the vertices of
 the minimum sized axis parallel square $\cal S$
 lies on $\overline{p_1p_2}$, $\overline{q_1q_2}$, $\overline{r_1r_2}$
and $\overline{s_1s_2}$. 
 It can be shown that, one can translate $\cal S$ ``horizontally towards left or right'',
 and/or ``vertically upward or downward'' keeping its size unchanged, without missing
  any segment (i.e. each segment remains hit by ${\cal S}$ always) to move  
  one of the vertices of $\cal S$ touching
 the respective segment.
\end{proof}

If there are multiple minimum sized congruent squares for hitting the segments
(See Figure~\ref{gruent_case_label}(a,b,b,d)), 
then our proposed algorithm for the {\bf LHIT problem} will  also
work. The reason is that after choosing an ${\cal S}_1$, our algorithm
for computing ${\cal S}_2$ needs only the segments that are not hit by ${\cal S}_1$.
We increase the size of ${\cal S}_1$ monotonically 
according to the event points 
corresponding to the top-right corner of ${\cal S}_1$.
Now in each step, if ${\cal S}_1$ hits a defining segment of ${\cal S}_2$,
then the size of ${\cal S}_2$ is reduced by 
eliminating that segment from it.
If there exists multiple congruent ${\cal S}_2$ of minimum size that hit all the segments which are not 
hit by ${\cal S}_1$, we can choose any one of them as square ${\cal S}_2$,
since all such ${\cal S}_2$'s are defined by the same subset
segments (Observation~\ref{basic}(i)).
\begin{figure}
\centering
\includegraphics[width=0.7\textwidth]{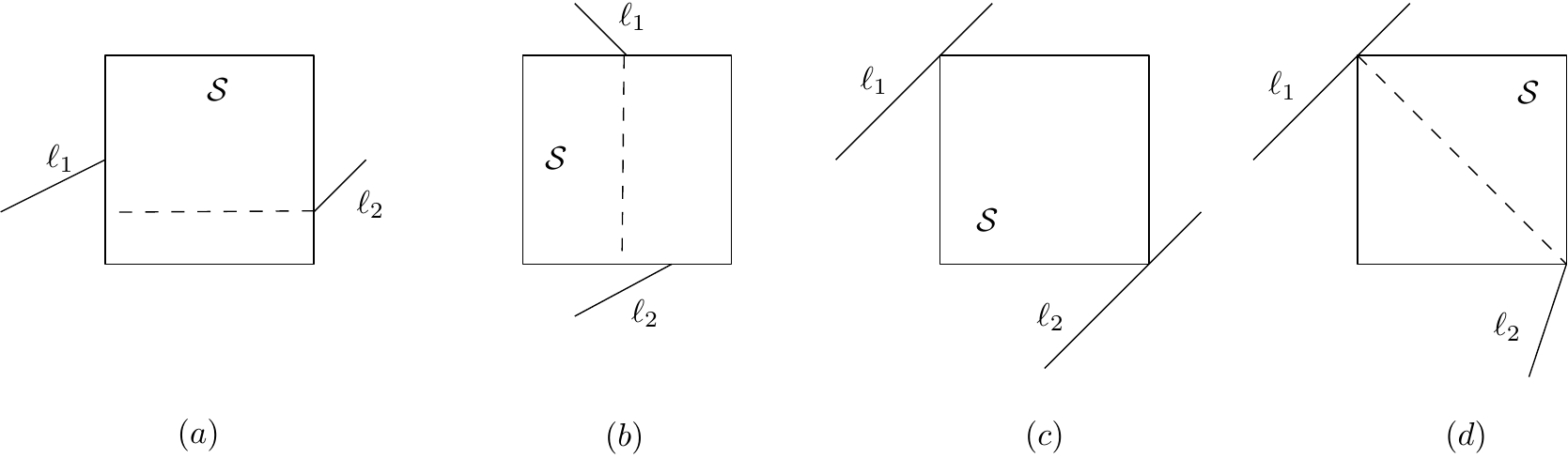}
\caption{Demonstration of multiple copies minimum sized square 
$\cal S$ defined by two segments $\ell_1$ and $\ell_2$:
(a) at the left and right boundary of $\cal S$ (b) at the top and bottom
boundary of $\cal S$ (c) at two diagonal vertices
of $\cal S$ where the segments are parallel, (d) at two diagonal vertices of $\cal S$
where the segments are non-parallel}
\label{gruent_case_label}
\end{figure}

\begin{lemma}
 \label{min_square}
 An axis parallel square of minimum size hitting all the members of
 a given set $\cal L$ of $n$ line segments can be obtained 
 in $O(n)$ time.
\end{lemma}
\begin{proof}
 Among the given set ${\cal L}$ of $n$ line segments, we can identify the 
 special line segments $\ell_i$, $i\in \{a,b,c,d,p,q,r,s\}$
 (see Figure~\ref{degenerated_fig1}) in $O(n)$ time.
 
 We now show that a minimum sized axis parallel square ${\cal S}^r$ whose ``top-right'' 
 corner lies on $\overline{r_1r_2}\in \ell_r$ and hits all the segments, can be computed 
 in $O(1)$ time. The same method works for computing the minimum sized squares ${\cal S}^p$,
 ${\cal S}^q$ and ${\cal S}^s$ whose one corner lies on $\overline{p_1p_2}$, 
 $\overline{q_1q_2}$ and $\overline{s_1s_2}$ respectively and hits all the line segments. 
 Finally we will choose minimum sized square among ${\cal S}^p$, ${\cal S}^q$, 
 ${\cal S}^r$ and ${\cal S}^s$. 
 
 {\bf Computation of ${\cal S}^r$}: For each $i\in\{a,~p,~q,~d,~s\}$, we compute 
 the locus $loc(i)$ of the ``bottom-left'' corner of a minimum sized square $\cal S$
 which hits the line segment $\ell_i$, while its ``top-right'' corner moving along 
 the segment $\overline{r_2r_1}$. In Figure~\ref{degenerated_fig11}(a), $loc(s)$ is 
 demonstrated, while in Figure~\ref{degenerated_fig11}(b) all the $loc(i)$, $i\in
 \{a,~p,~q,~d,~s\}$ are shown. We also compute the locus of the ``bottom-left''
 corner of $\cal S$ (denoted by $loc(b,c)$ in Figure \ref{degenerated_fig11}(b)) 
 that hits both $\ell_b$ and $\ell_c$ while the top-right corner of $S$ moves along
 the segment $\overline{r_2r_1}$. Each of the locii in $\{loc(i),i = a,p,q,d,s,(b,c)\}$ 
 consists of at most three line segments (see Appendix
 for details).
 We consider two lines $DL_1$ and $DL_2$ of unit slope passing through $r_1$ and $r_2$ 
 respectively (see Figure~\ref{degenerated_fig11}(b)). We can compute the upper envelope $\large U$
 (as the distance is measured from $\overline{r_2r_1}$) of the 
 locii $\{loc(i)$, $i\in\{a,p,q,d,s,(b,c)\}\}$ within the strip bounded by $DL_1$ and 
 $DL_2$ (colored red in Figure~\ref{degenerated_fig11}(b)) in $O(1)$ time. The square 
 whose ``bottom-left'' corner lies on the upper envelope $\large U$ while its ``top-right'' corner lies on $\overline{r_2r_1}$, hits all the segments $\ell_i$, 
 $i\in\{a,b,c,d,p,q,r,s\}$. Thus, the upper envelope $\large U$
 corresponds to the locus of the 
 bottom-left corner of ${\cal S}^r$ that hits all the segment in $\cal L$ (see 
 Fact~\ref{simple}) while its top-right corner moves along $\overline{r_2r_1}$. 
 Note that $\large U$ consists of a constant number of segments and it can be computed 
 in $O(1)$ time. As one moves along an edge of $\large U$, the 
 size of the square ${\cal S}^r$ either monotonically increases or decreases or 
 remains same. So, the minimum size of the square ${\cal S}^r$ occurs at some vertex 
 of $\large U$, and it can be determined by inspecting all the vertices of $\large U$. 
 
 If any one of $\ell_p$, $\ell_q$, $\ell_r$ and $\ell_s$ does not exist in the given 
 instance with the segments $\cal L$, then the corresponding locus is not present, 
 and the  same method works in such a situation with the available set of locii. 
 \end{proof}

\section{Two hitting squares} 
We now discuss the hitting problem by two axis parallel squares (${\cal S}_1$, 
${\cal S}_2$) using the method described in Section \ref{one-hit} as a subroutine. We assume 
that ${\cal S}_1$ hits $\ell_p$ along with some other members in $\cal L$. 
${\cal S}_2$ must hit the members that are not hit by ${\cal S}_1$. 
Our objective is to compute the pair (${\cal S}_1$, ${\cal S}_2$) that 
minimizes $\max(size({\cal S}_1),size({\cal S}_2))$. 
 
\begin{lemma}
\label{cornerx}
To minimize the $\max(size({\cal S}_1),size({\cal S}_2))$,  
the ``bottom-left'' corner of ${\cal S}_1$ will lie on $\ell_p$.
\end{lemma}
\begin{proof}
Suppose ${\cal L}_1 \subset {\cal L}$ be the set of segments hit by ${\cal S}_1$
when $\max(size({\cal S}_1),size({\cal S}_2))$ is minimized.
Let the ``bottom-left'' corner of ${\cal S}_1$ lie below $\ell_p$ i.e. both bottom boundary
and left boundary of ${\cal S}_1$ properly intersect $\ell_p$ (see Figure~\ref{degenerated_fig12}). 
Let $\ell_1, \ell_2\in {\cal L}_1$ be two segments so that the $y$-coordinate
(resp. $x$-coordinate) of top end-point (resp. right end-point) of $\ell_1$
(resp. $\ell_2$) is minimum among that of all the segment $\ell_k\in {\cal L}_1$.
If the bottom (resp. left) boundary of ${\cal S}_1$ properly intersect $\ell_1$ (resp. $\ell_2$), 
we can translate ${\cal S}_1$ vertically upwards (resp. horizontally rightwards) keeping its size same,
so that the bottom boundary (resp. left boundary) of  ${\cal S}_1$ touches
$\ell_1$ (resp. $\ell_2$) or the bottom-left corner of ${\cal S}_1$
touches $\ell_p$.
If $\ell_p$ is touched, the result is justified. If $\ell_1$ (resp. 
$\ell_2$) is touched, we can translate ${\cal S}_1$ towards right (resp. above) to make the 
bottom-left corner of ${\cal S}_1$ touching $\ell_p$. 
The revised ${\cal S}_1$ also hits all the segments in ${\cal L}_1$.
\end{proof}

\begin{figure}
\centering
\includegraphics[width=0.7\textwidth]{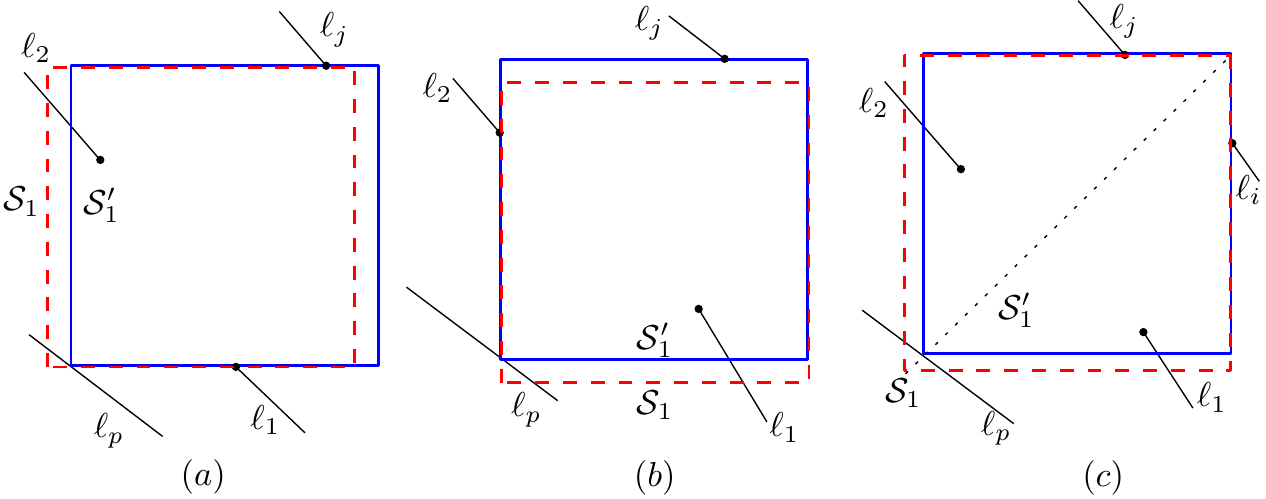}
\caption{Proof of Lemma \ref{cornerx}}
\label{degenerated_fig12}
\end{figure}

Lemma~\ref{cornerx} says that a square $\cal S$ serves as ${\cal S}_1$ if the boundary of $\cal S$ touches $\ell_p$ and also hits
a subset ${\cal L}' \subset {\cal L}\setminus \{\ell_p\}$ with at least one
segment of ${\cal L}'$ touching the boundary
of $\cal S$ from outside.
 The reason of defining ${\cal S}_1$ in such a manner is that if all the segments ${\cal L}'$ hit by ${\cal S}_1$ lie 
either inside ${\cal S}_1$ or properly intersect the boundary of ${\cal S}_1$, then we can reduce the size of ${\cal S}_1$ hitting 
the same set of segments. Now, we will introduce the concept of defining ${\cal S}_1$ using a subset of $\cal L$ as follows:

\begin{definition}
A subset ${\cal L}' \subseteq {\cal L}\setminus \{\ell_p\}$ is said to be {\em minimal}
to define a square $\cal S$ 
(with bottom-left corner is on $\ell_p$) as ${\cal S}_1$
if the members of ${\cal L}'$ uniquely determine its top-right corner of $\cal S$, and no proper subset of ${\cal L}'$ can define 
the top-right corner of $\cal S$ uniquely.
\end{definition}

We will consider possible subsets ${\cal L}_1 \subset \cal L$ that
can define ${\cal S}_1$, and invoke the procedure 
described in Section \ref{min_square} with the subset ${\cal L}\setminus 
\left({\cal L}_1\cup\{\ell_p\}\right)$ to compute ${\cal S}_2$. The following
Lemma \ref{lx} and Lemma \ref{lxx} says that we need to consider the two cases separately 
depending on whether the bottom-left 
corner of ${\cal S}_1$, denoted by $\pi$, resides at (i) an end-point of $\ell_p$, 
and (ii) an intermediate point of $\ell_p$.

\begin{lemma} \label{lx}
If $\pi$ coincides with an end-point of $\ell_p$ (Case (i)), then ${\cal S}_1$ is determined by a 
single segment of ${\cal L}\setminus \{\ell_p\}$. 
\end{lemma}

\begin{proof}
Here, the top-right corner $\pi'$ of ${\cal S}_1$ lies on a line of unit slope 
 passing through $\pi$. We need to investigate the following three
 exhaustive cases. 
\begin{itemize}
\item $\pi'$ lies on a segment $\ell_i \in {\cal L}\setminus \{\ell_p\}$ 
(see Figure~\ref{degenerated_fig3}(b)), or 
\item  $\pi'$ lies on the vertical line passing through the left end-point of 
a segment $\ell_i \in {\cal L}\setminus \{\ell_p\}$ 
(see Figure~\ref{degenerated_fig3}(c, d)), or  
\item  $\pi'$ lies on the horizontal line passing through the bottom end-point of a segment
$\ell_i \in {\cal L}\setminus \{\ell_p\}$ 
(see Figure~\ref{degenerated_fig3}(a, e)).  
\end{itemize}
This is due to the fact that if none of these cases happen then we can get 
another square, say ${\cal S}_1'$, of reduced size whose bottom-left corner is 
at $\pi$  and it hits all the segments in $\cal L$ that are also hit by 
${\cal S}_1$. Here ${\cal S}_1'$ serves the purpose of ${\cal S}_1$.
Thus,the lemma follows. 
\end{proof}

\begin{lemma}\label{lxx}
If $\pi$ coincides with an intermediate point of $\ell_p$ (Case (ii)),
then ${\cal S}_1$ is determined by two  
segment of ${\cal L}\setminus \{\ell_p\}$. 
\end{lemma}
\begin{proof}
In this case, the bottom-left corner of ${\cal S}_1$ will be determined as follows: 
\begin{itemize}
\item[$\bullet$] a segment $\ell_i \in {\cal L}\setminus \{\ell_p\}$ defines 
the bottom boundary of ${\cal S}_1$ whose horizontal projection $\pi$ on 
$\ell_p$ determines the bottom-left corner of ${\cal S}_1$
(see Figure~\ref{degenerated_fig4}(d,~e)), or 
\item[$\bullet$] a segment $\ell_i \in {\cal L}\setminus \{\ell_p\}$ defines 
the left boundary of ${\cal S}_1$ whose vertical projection $\pi$ on 
$\ell_p$ determines the bottom-left corner of ${\cal S}_1$
(see Figure~\ref{degenerated_fig4}(a,~b)), or 
\item[$\bullet$] a pair of segments $\ell_i$ and $\ell_i'$ defines 
the top-right corner $\pi'$ of ${\cal S}_1$, and 
the point of intersection of a line of unit slope passing through $\pi'$  
with the line segment $\ell_p$
determines the bottom-left corner of ${\cal S}_1$ (see Figure~\ref{degenerated_fig4}(c)).
\end{itemize}
In the first and second bulleted case, Lemma \ref{lx} says that one more segment $\ell_j$ is required
to define the top-right corner of ${\cal S}_1$. In the third bulleted 
case, both the bottom-left and the top-right corners of ${\cal S}_1$
are already defined. Thus, the lemma follows.
\end{proof}

In the following two subsections we will compute ${\cal S}_1$ considering the two cases where 
(i) ${\cal S}_1$ is defined by one segment in ${\cal L}\setminus \{\ell_p\}$ and 
(ii) two segments in ${\cal L}\setminus \{\ell_p\}$ respectively. Note that, if a single segment 
$\ell \in \cal L$ touches a corner of ${\cal S}_1$, then $\ell$ is said to touch both 
the boundaries of ${\cal S}_1$ adjacent to that corner (see Figure~\ref{degenerated_fig4}(f)).

\begin{figure}[h]
\centering
\includegraphics[width=0.7\textwidth]{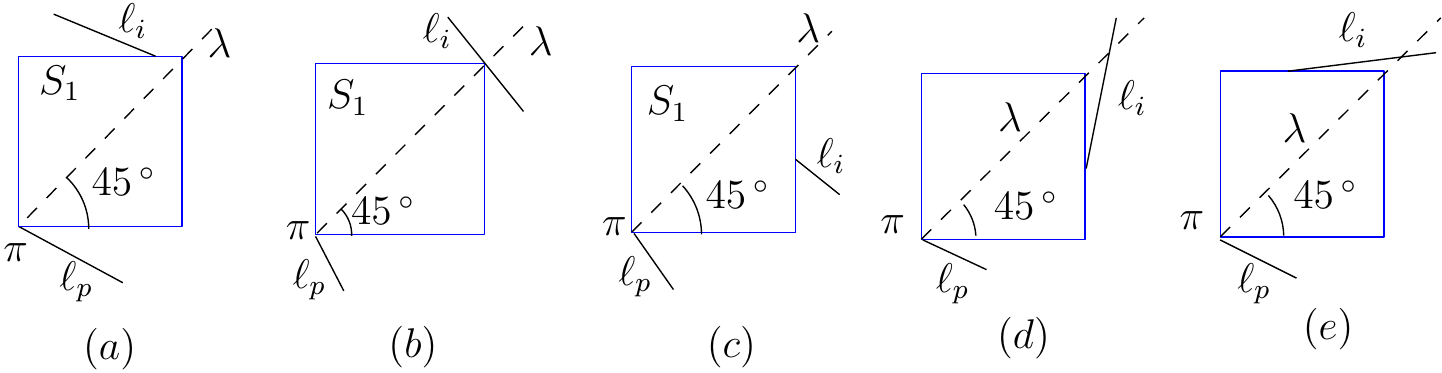}
\caption{The ``bottom-right'' corner of square ${\cal S}_1$ is at a 
segment end-point}
\label{degenerated_fig3}
\end{figure} 

\subsubsection*{(A) ${\cal S}_1$  is defined by one line segment:}

We draw a straight line $\lambda$ of slope ``1'' through an end-point $\pi$ of 
$\ell_p$. Next, we consider each segment 
 $\ell_i \in {\cal L}\setminus \{\ell_p\}$, and create an array $Q$ of event points as follows:  
\begin{itemize}
\item If $\ell_i$ is strictly above $\lambda$ (Figure~\ref{degenerated_fig3}(a)),
store the horizontal 
projection $q$ of the bottom end-point of $\ell_i$ on the line $\lambda$ in $Q$. 
\item If $\ell_i$ with negative slope intersects $\lambda$ at a point $q$ 
(Figure~\ref{degenerated_fig3}(b)), we store $q$ 
 in $Q$.
\item If $\ell_i$ with positive slope ($\leq 1$) intersects $\lambda$  
(Figure~\ref{degenerated_fig3}(e)), store the horizontal 
projection $q$ of the bottom end-point of $\ell_i$ on the line $\lambda$ in $Q$. 
\item If $\ell_i$ with positive slope ($> 1$) intersects $\lambda$ 
(Figure~\ref{degenerated_fig3}(d)), store the vertical 
projection $q$ of the left end-point of $\ell_i$ on the line $\lambda$ in $Q$. 
\item If  $\ell_i$ is strictly below $\lambda$ 
(Figure~\ref{degenerated_fig3}(c)), then store the vertical 
projection $q$ of the left end-point of $\ell_i$ on $\lambda$ in $Q$. 
\end{itemize}

We consider each member $q \in Q$. Define ${\cal S}_1$ with its 
(bottom-left, top-right) corner points as $(\pi,q)$. Identify the 
subset ${\cal L}_1$ of segments in 
${\cal L}$ that are hit by ${\cal S}_1$. Call the 
procedure of Section \ref{min_square} with the set of segments 
${\cal L}\setminus{\cal L}_1$ to compute ${\cal S}_2$. Replace the current optimum square-pair 
by $\max(size({\cal S}_1),size({\cal S}_2))$ if needed.

\begin{lemma} 
\label{one}
The minimum of the size of the optimum pair of squares where  
${\cal S}_1$  is defined by one line segment of ${\cal L} \setminus \{\ell_p\}$
can be computed in $O(n^2)$ time.
\end{lemma}
\begin{proof}
The array $Q$ can be computed in $O(n)$ time. For each member 
$q \in Q$, (i) the subset ${\cal L}_1$ of $\cal L$ can be identified in $O(n)$ 
time, and then (ii) the time required for computing ${\cal S}_2$ is also $O(n)$.
As $|Q|=O(n)$, the result follows.  
\end{proof}

\subsubsection*{(B) The top-right corner of ${\cal S}_1$  is defined by two line segments :}

By Lemma \ref{lxx}, assuming that the bottom-left corner of ${\cal S}_1$ lies in the interior of $\ell_p$, 
we need to consider the following  cases to uniquely define the 
possible bottom-left corner of ${\cal S}_1$.

\begin{itemize}
\itemsep-0.5em 
 \item[B1:] The bottom-left corner of ${\cal S}_1$ is defined by 
 the top end-point of a segment $\ell_i$ touching its bottom boundary
 (see Figure~\ref{degenerated_fig4}(d,~e)).
 \item[B2:] The bottom-left corner of ${\cal S}_1$ is defined by the
 right end-point of a segment $\ell_i$ touching its left boundary
 (see Figure~\ref{degenerated_fig4}(a,~b)).
 \item[B3:] The bottom-left corner of ${\cal S}_1$ is defined by its top-right corner $\pi'$,  
 defined by a pair of segments $\ell_i$
 and $\ell_j$ touching the ``top'' and ``right'' boundaries of ${\cal S}_1$
 (see Figure~\ref{degenerated_fig4}(c)).
\end{itemize}

\begin{figure}[b]
\centering
\includegraphics[width=.98\textwidth]{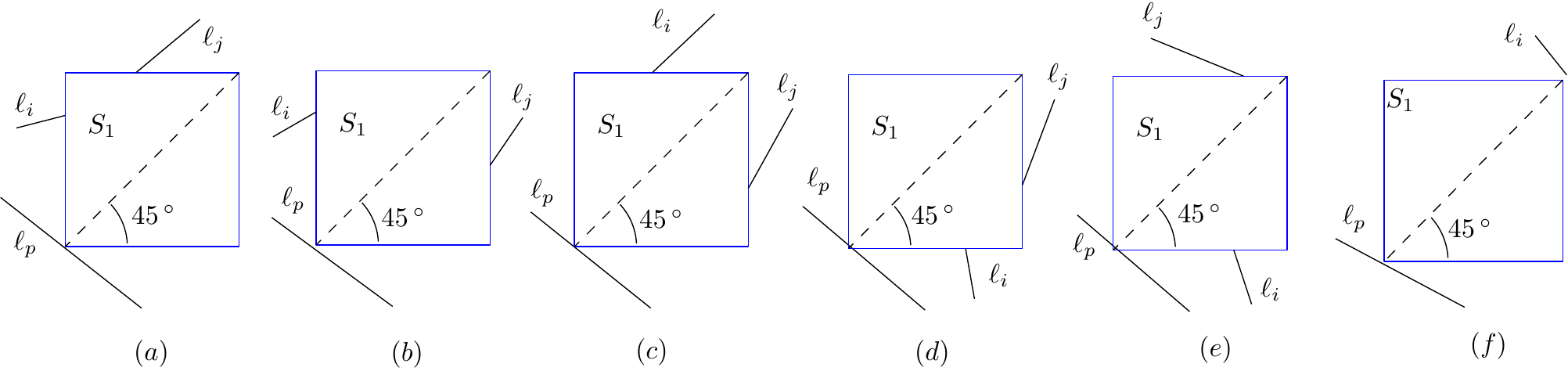}
\caption{The ``top-right'' corner of ${\cal S}_1$ that hits $\ell_p$ is defined by two
 segments $\ell_i$ and $\ell_j$.}
\label{degenerated_fig4}
\end{figure}
Note that, Figure ~\ref{degenerated_fig4}(f) is basically the case B3, where $\ell_i$ is assumed to touch 
both the ``top'' and ``right'' boundaries of ${\cal S}_1$.

We use four arrays ${\cal L}_l$, ${\cal L}_r$, ${\cal L}_t$ and ${\cal L}_b$, each with 
the members in ${\cal L}$ sorted with respect to their left, right, top, 
and bottom end-points respectively. In addition, we keep a sorted array ${\cal L}_d$ containing 
the points of intersection of the line containing $\ell_p$ and the lines of slope 1 (called diagonal 
lines) at both the end-points of each member in ${\cal L}\setminus \{\ell_p\}$.
Each element $\ell_i \in \cal L$ maintains 
six pointers to the corresponding element in ${\cal L}_l$, ${\cal L}_r$, ${\cal L}_t$, ${\cal L}_b$ 
and to two elements of ${\cal L}_d$ corresponding to its two end-points. Also, each element of 
${\cal L}_i$, $i=l,r,t,b,d$ points to the corresponding segment $\ell \in \cal L$. In addition, 
we also maintain four ordered arrays, namely ${\cal I}^{v1}(\tau)$, ${\cal I}^{v2}(\tau)$
${\cal I}^h(\tau)$ and ${\cal I}^d(\tau)$ 
for each end-point $\tau$ of the members in $\cal L$.
${\cal I}^{v1}(\tau)$ (resp. ${\cal I}^{v2}(\tau)$) is the 
list of segments hit by an upward (resp. downward) vertical ray from $\tau$, and  
${\cal I}^h(\tau)$ (resp. ${\cal I}^d(\tau)$) is the list of segments in
$\cal L$ intersected by the {\it horizontal line} (resp. {\it diagonal line}) passing 
through the point $\tau$ in sorted order. Each segment 
$\ell_i \in \cal L$ maintains eight pointers to point the lists 
${\cal I}^{v1}(\tau)$, ${\cal I}^{v2}(\tau)$, ${\cal I}^{h}(\tau)$, 
${\cal I}^d(\tau)$, ${\cal I}^{v1}(\tau')$, ${\cal I}^{v2}(\tau')$,
${\cal I}^{h}(\tau')$ and ${\cal I}^d(\tau')$ where $\tau$ and 
$\tau'$ are two end-points of $\ell_i$. The arrays ${\cal L}_i$, $i=l,r,t,b,d$ can be created in 
$O(n\log n)$ time. Also, the arrays ${\cal I}^{v1}(\tau)$, ${\cal I}^{v2}(\tau)$, 
${\cal I}^{h}(\tau)$ and ${\cal I}^d(\tau)$ 
for all the $2n$ end-points ($\tau$) of the segments in $\cal L$ can be created in $O(n^2)$ time and 
will be stored using $O(n^2)$ space.

Let us now consider the generation of the instances in B1. Lemma \ref{cornerx} says that if $\ell_p$ 
exists, then the bottom-left corner of ${\cal S}_1$ lies on $\ell_p$. We first generate all possible 
bottom-left corners $\cal C$ of ${\cal S}_1$ on $\ell_p$ in sorted order whose bottom boundary is 
supported by the top end-point of a segment $\ell_i$ in $\cal L$ by traversing the list ${\cal L}_t$. For 
each element $\theta \in \cal C$ (correspnding to the top-end point of a line segment $\ell_i$), 
we consider a half-line $\lambda(\theta)$ of slope ``1'' at the point $\theta$, and generate the 
array ${\cal D}_\theta$ that contains the top-right corner of all possible squares ${\cal S}_1$ lying 
on $\lambda(\theta)$, in order of their distances from the point $\theta$ 
(see Figure~\ref{degenerate_D_theta}). We denote the horizontal line at 
$\theta$ by $h_\theta$. The elements (known as event points)
of the array ${\cal D}_\theta$ are the points of intersection of $\lambda(\theta)$ with 
\begin{itemize}
 \item[(i)] the vertical lines at the {\it left end-point} of all the segments in $\cal L$ whose left end-point lies 
 below the line $\lambda(\theta)$  and above the
 line $h_\theta$ (see {\it red} points e.g. $e^4_i$, $e^5_i$, $e^6_i$ in Figure~\ref{degenerate_D_theta}), 
 \item[(ii)] the vertical lines at the point of intersection of $h_\theta$ with the segments ${\cal L}' \subseteq \cal L$, provided the slope of the segments in ${\cal L}'$ are positive (see {\it blue} 
 points e.g. $e^1_i$ in Figure \ref{degenerate_D_theta}), 
 \item[(iii)] the horizontal line at the bottom end-point of all the segments whose bottom end-point lies above 
 $\lambda(\theta)$ (see {\it green} points e.g. $e^3_i$, $e^8_i$, $e^9_i$ in Figure \ref{degenerate_D_theta}), and  
 \item[(iv)] the segments in $\cal L$ with negative slope that intersects $\lambda(\theta)$ (see {\it pink} points $e^2_i$ in Figure~\ref{degenerate_D_theta}),
\end{itemize}

\begin{figure}
\centering
 \includegraphics[width=.5\textwidth]{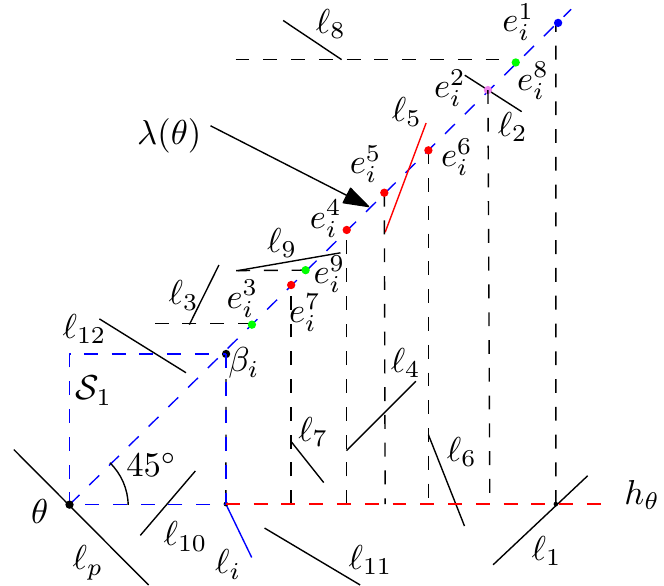}
\caption{Generation of ${\cal D}_\theta$ where $\theta$ is horizontal
projection of top end-point of $\ell_i$ on $\ell_p$}
\label{degenerate_D_theta}
 
\end{figure}
%
Since ${\cal S}_1$ hits $\ell_i$, we need to remove all the events generated on $\lambda(\theta)$ 
whose $x$-coordinates are less than that of the top end-point $\tau$
of $\ell_i$ (e.g. events for $\ell_{10}$, $\ell_{12}$
in Figure~\ref{degenerate_D_theta}). 

The Type (i) (resp. Type (iii)) events are generated in increasing order of their $x$-coordinates by 
scanning the array ${\cal L}_l$ (resp. ${\cal L}_b$). Type (ii) events are created in 
increasing order of $x$-coordinates from the list ${\cal I}^h(\tau)$, where
the horizontal projection of the top end-point $\tau$ of the line segment $\ell_i$
on $\ell_p$ is $\theta$. Type~(iv) events are identified from the two ordered 
arrays ${\cal I}^d(p_1)$ and 
${\cal I}^d(p_2)$ where $p_1$ and $p_2$ are two end-points
of (same or different) line segments that generated two consecutive event points $e$ and $e'$
in the array ${\cal L}_d$, and
$x(e) \leq x(\theta) \leq x(e')$.
Note that we need to consider only the
segments of negative slope 
in ${\cal I}^d(p_1)\cup{\cal I}^d(p_2)$ in ordered manner to compute Type~(iv).

Now, we merge the events of Types (i) to (iv) to get the list ${\cal D}_\theta$ containing all possible 
events on $\lambda_\theta$ arranged in increasing order of their $x$-coordinates. We process each event 
of  $\delta \in {\cal D}_\theta$ by executing the steps (i) compute an ${\cal S}_1$ square 
with (bottom-left, top-right) corners at $(\theta,\delta)$, (ii) identify the 
segments in ${\cal L}' \subseteq {\cal L}$ that are hit by ${\cal S}_1$, and (iii) for the remaining segments 
${\cal L}\setminus {\cal L}'$, we compute ${\cal S}_2$ in $O(1)$ amortized time as described below.
\begin{figure}
\centering
 \includegraphics[width=.9\textwidth]{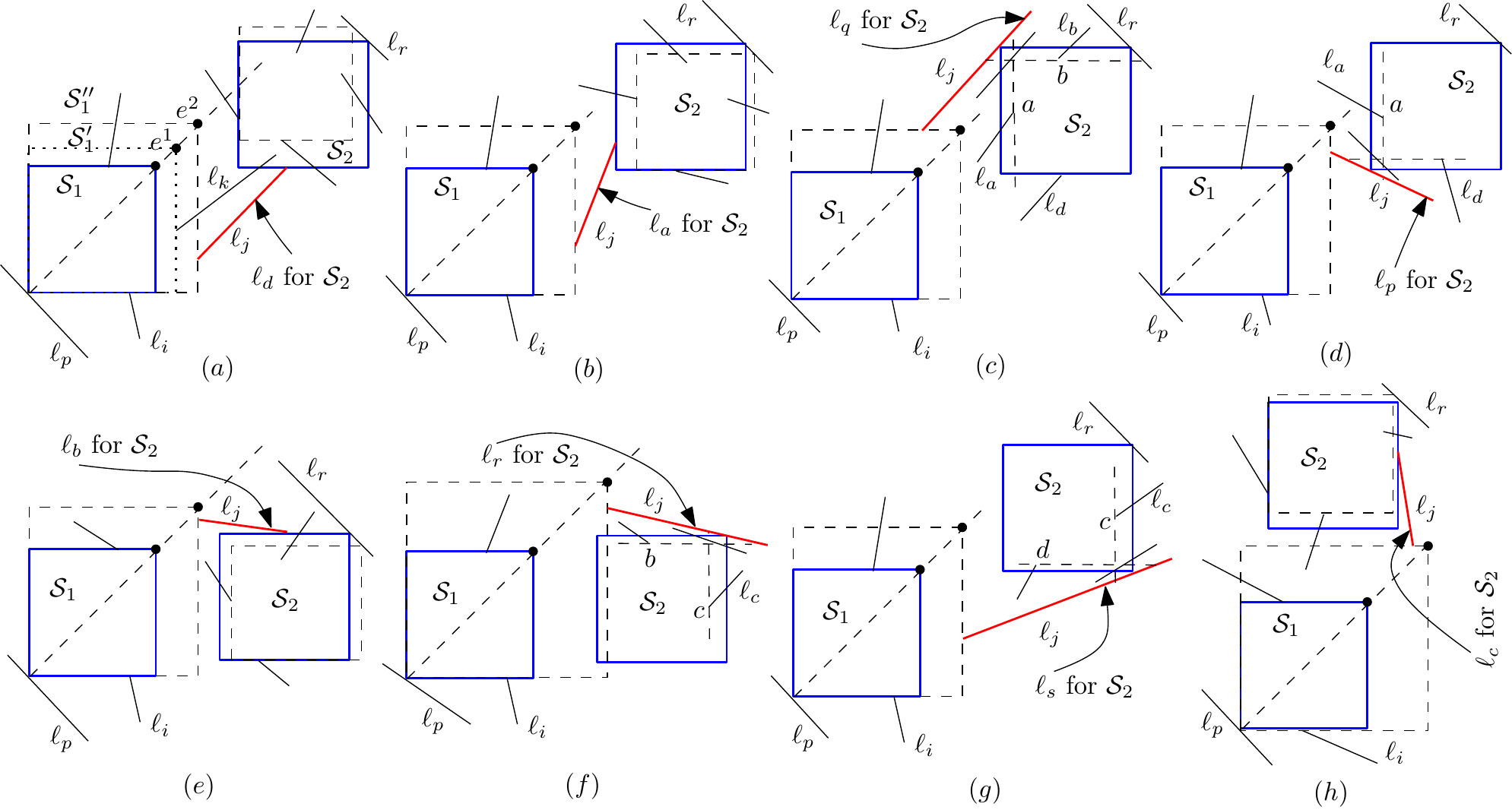}
\caption{Demonstration of Iterative steps of computing ${\cal S}_2$ for different elements of ${\cal D}_\theta$}
\label{x} 
\end{figure}

\begin{description}
 \item[Initialization step:] For the first event $\delta_1 \in  {\cal D}_\theta$, we apply the algorithm of Section 
 \ref{one-hit} to compute ${\cal S}_2$. This also identifies the segments $\ell_a, \ell_b, 
 \ell_c,\ell_d,\ell_p, \ell_q, \ell_r, \ell_s \in {\cal L}\setminus {\cal L}'$ 
 as defined in Lemma \ref{simple}. This needs $O(n)$ time.
 \item[Iterative step:] Below, we show that, after processing $\delta_i \in  {\cal D}_\theta$, 
 when we process 
 $\delta_{i+1}~\in~{\cal D}_\theta$ in order, at
 most one among the eight segments $\ell_a, \ell_b, 
 \ell_c,\ell_d,\ell_p, \ell_q, \ell_r, \ell_s \in {\cal L}\setminus {\cal L}'$ for ${\cal S}_2$ (see the eight situations 
 in Figure \ref{x}), may change, and it can be obtained in $O(1)$ time. 
\begin{description}
\item[] In Figure~\ref{x}(a), if ${\cal S}_1$ is increased to ${\cal S}'_1$ (dotted square),
then none of the $8$ segments of ${\cal S}_2$ gets changed.	
\item[] In Figure \ref{x}(a), if ${\cal S}_1$ is increased to ${\cal S}''_1$ (dashed square),
then $\ell_d$ of ${\cal S}_2$ gets changed,
which can be obtained by scanning ${\cal L}_t$ array.
\item[] In Figure \ref{x}(b)  $\ell_a$ of ${\cal S}_2$ gets changed, 
which can be obtained by scanning ${\cal L}_r$ array.
\item[] In Figure \ref{x}(c)  $\ell_q$ of ${\cal S}_2$ gets changed, 
which can be obtained by scanning ${\cal I}^{v1}(a)$ array.
\item[] In Figure \ref{x}(d)  $\ell_p$ of ${\cal S}_2$ gets changed,
which can be obtained by scanning ${\cal I}^{v2}(a)$ array.
\item[] In Figure \ref{x}(e)  $\ell_b$ of ${\cal S}_2$ gets changed, 
which can be obtained by scanning ${\cal L}_b$ array.
\item[] In Figure \ref{x}(f)  $\ell_r$ of ${\cal S}_2$ gets changed, 
which can be obtained by scanning ${\cal I}^{v1}(c)$  array.
\item[] In Figure \ref{x}(g)  $\ell_s$ of ${\cal S}_2$ gets changed, 
which can be obtained by scanning ${\cal I}^{v2}(c)$  array.
\item[] In Figure \ref{x}(h)  $\ell_c$ of ${\cal S}_2$ gets changed, 
which can be obtained by scanning ${\cal L}_l$ array.
\end{description}

\end{description}
\vspace{-0.1in}
The processing of all the elements in ${\cal D}_\theta$ needs exactly
one scan of the arrays ${\cal L}_b$, ${\cal L}_r$, ${\cal L}_t$,
${\cal L}_l$, ${\cal I}^{v1}(\tau)$, ${\cal I}^{v2}(\tau)$, ${\cal I}^h(\tau)$, 
${\cal I}^d(\tau)$, ${\cal I}^{v1}(\tau')$, ${\cal I}^{v2}(\tau')$. Thus, we can  compute the 
required ${\cal S}_2$ for each element in $\delta \in {\cal D}_\theta$ in amortized $O(1)$ time.
 The generation of the instances in B2 are similar to that of B1.  
To generate the instances of B3 with the segment $\ell_j$
on its right boundary, we need to consider a vertical line
$V_j$ at the left end-point on $\ell_j$, and 
include the horizontal projection of the bottom end-point 
of all the segments in ${\cal L} \setminus \{ \ell_p \} $ on $V_j$ provided the
concerned bottom end-points lie to the left of $V_j$
and above the left end-point of $\ell_j$. 
For all the segments in $\cal L$  with negative slope that 
intersects $V_j$ above the left end-point of $\ell_j$, we include those points of intersection 
in $V_j$. We also include the left end-point of $\ell_j$
as an event in $V_j$. These events can be generated in $O(n)$ 
time using the array ${\cal L}_b$. For each of these events the corresponding 
${\cal S}_1$ square and hence the correspnding ${\cal S}_2$ square are well-defined. The 
${\cal S}_2$ squares for all the events in $V_j$ can also be computed in $O(n)$ time. 
Thus, we have the following theorem:

\begin{theorem}
 If ${\cal R}_{abcd}$ does not hit all the line segments in $\cal L$,
 we can compute the optimal axis parallel square pair (${\cal S}_1$, ${\cal S}_2$)
 that combinedly hit all the segments in $\cal L$ in $O(n^2)$ time.
\end{theorem}

\begin{proof}
Lemma \ref{one} says that if the ${\cal S}_1$ square is
defined by one line segment in ${\cal L}\setminus \{\ell_p\}$,
we can compute the optimum pair of squares $({\cal S}_1,{\cal S}_2)$ 
in $O(n^2)$ time. The instances where  ${\cal S}_1$ is defined by two
line segments in ${\cal L}\setminus \{\ell_p\}$, are classified
into three cases B1, B2, B3. 
For handling the case B1, we created $O(n)$ events on $\ell_p$ 
in the array $\cal C$ in $O(n)$ time using the  ${\cal L}_t$ array.
These corresponds to the bottom left corner of possible ${\cal S}_1$.
For each event $\theta \in C$, we create another array ${\cal D}_\theta$ with $O(n)$ sub-events each 
of them may be the top-right corners of ${\cal S}_1$ square 
whose bottom-left corner is $\theta$. 
We can process these $O(n)$ events in ${\cal D}_\theta$ in amortized $O(n)$ time.
Thus, all possible instances of type B1 can be generated in $O(n^2)$ time. 
Similarly, all possible instances of type B2 also can be generated in 
$O(n^2)$ time. Regarding the instances of type B3, we need to consider the 
left end-points of all the $O(n)$ segments in $\cal L$. As mentioned 
earlier, the number of events (top-right corner of ${\cal S}_1$ squares) generated 
is $O(n)$, and they can be processed in amortized $O(n)$ time. 
In special case of B3 (see Figure~\ref{degenerated_fig4}(f)), both the top and right 
boundaries of the square ${\cal S}_1$
is touched by a segment $\ell_i$,
and the correspnding ${\cal S}_2$ can be determined in $O(n)$
time. Since there are $n$ such line segments $\ell_i\in {\cal L}$, the total 
time complexity result for identifying all such instances is also $O(n^2)$.
Thus the result follows.
\end{proof}

\appendix
\begin{figure}[h]
\centering
\vspace{-0.2in}
\includegraphics[width=0.9\textwidth]{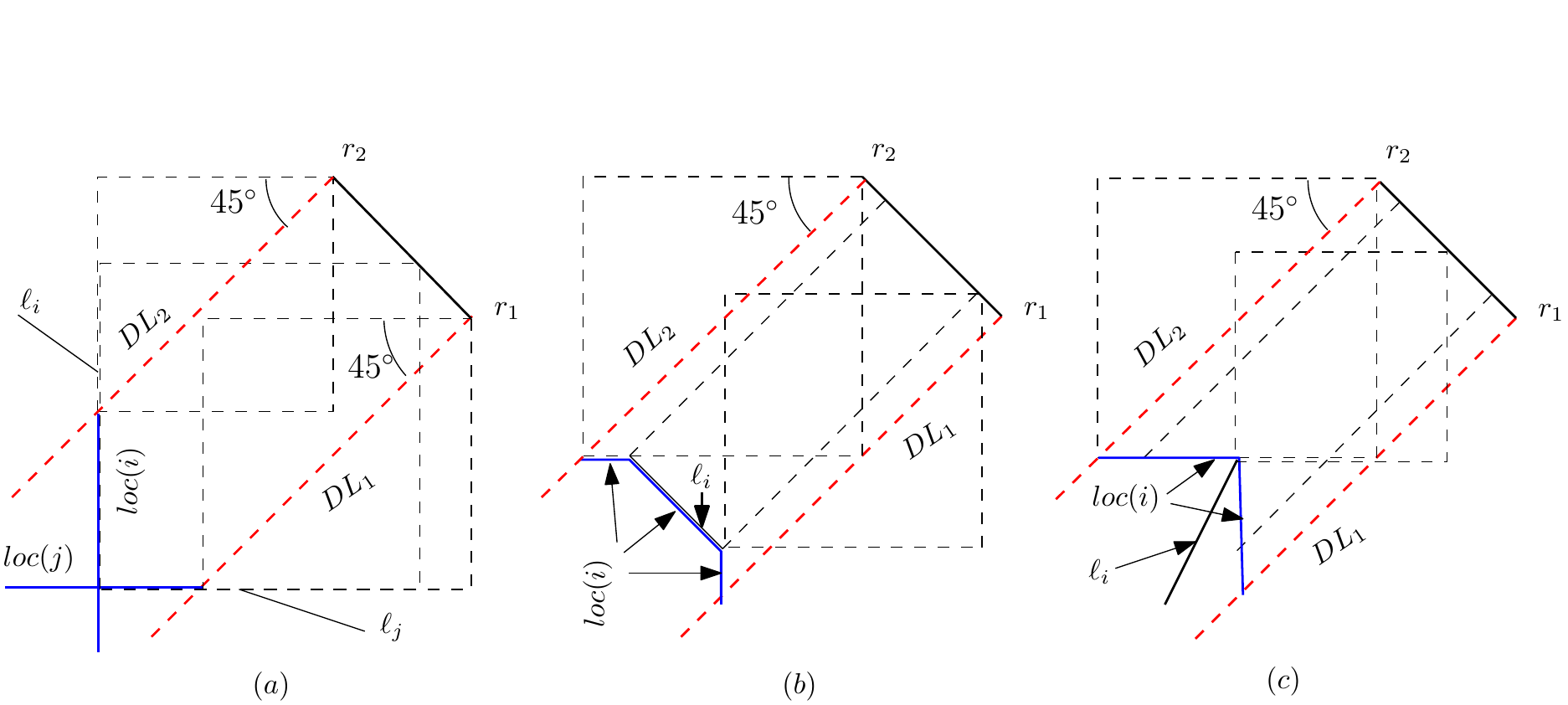}
\caption{The locus $loc(i)$ of the bottom left corner of square  that hits the 
segment $l_i$.}
\label{appendix_fig}
\end{figure}
{\bf Size of (i.e. the number of segments in) \boldmath $loc(i),~i = \{a,p,q,d,s,(b,c)\}$:}

The $loc(i)$ is the locus of the ``bottom-left'' corner of a minimum sized square ${\cal S}^r$
 which hits the line segment $\ell_i$, while its ``top-right'' corner moves along 
 the segment $\overline{r_2r_1}$ (The Figure~\ref{degenerated_fig11}(a) demonstrates $loc(s)$).
The $loc(i)$ (within the strip $\Gamma$ bounded by 
the line $DL_2$ and $DL_1$ of unit slope passing through $r_2$ and $r_1$ respectively)
is as follows:
\begin{itemize}
 \item If the segment $\ell_i$ (resp. $\ell_j$) lies above $DL_2$
 (resp. below $DL_1$),
then the required locus will be a vertical line (resp. horizontal line) inside the strip $\Gamma$ (see Figure~\ref{appendix_fig}(a)).
 \item If $\ell_i$ lies inside the strip $\Gamma$,
 then there are two possiblities:\\
(a) Slope of $\ell_i$ is negative (see Figure~\ref{appendix_fig}(b)): The 
required locus will be a horizontal segment 
passing through the top end-point of $\ell_i$ (to the left of it), 
until the bottom-left corner of the square coincides with the top end-point of $\ell_i$;
then it will move along $\ell_i$ till the bottom end-point
of $\ell_i$ is reached, and finally it will be vertically downwards,
until it hits the boundary of $\Gamma$.\\
(b) Slope of $\ell_i$ is positive (see Figure~\ref{appendix_fig}(c)): The 
required locus will be a horizontal segment as in case (a)
until the bottom-left corner of square hits the top end-point of $\ell_i$,
then finally it will be vertically downwards, until the boundary of $\Gamma$ is hit.
\item If $\ell_i$ intersects the boundary of $\Gamma$,
then also we can construct the required locus in a similar way as in the aforesaid cases.
\end{itemize}
Thus, in all the situations $loc(i)$ consists of at most three segments within $\Gamma$,
where at most one of them is non-axis-parallel.

\vspace{-0.1in}
\begin{thebibliography}{99}
\bibitem{sadhu} 
S. Sadhu, S. Roy, S. C. Nandy, and  S. Roy, \newblock  
Linear time algorithm to cover and hit a set of line segments
optimally by two axis-parallel squares, \newblock {\sf Theoretical Computer Science}, 769, pages 63--74, 2019.
\end{thebibliography}
\end{document}